\documentclass[conference,onecolumn]{IEEEtran}
\IEEEoverridecommandlockouts

\usepackage{todonotes}
\makeatletter
  \renewcommand{\todo}[2][]{\@todo[inline,color=yellow!20,#1]{#2}\ignorespaces}
\makeatother

\usepackage{comment}

\newtheorem{theorem}{Theorem}
\newtheorem{remark}{Remark}

\newtheorem{proposition}{Proposition}
\newtheorem{example}{Example}

\usetikzlibrary{arrows.meta, calc}
\definecolor{customblue}{RGB}{0, 0, 0}

\usepackage{cite}
\usepackage{amsmath,amssymb,amsfonts}
\usepackage{algorithmic}
\usepackage{graphicx}
\usepackage{textcomp}
\usepackage{xcolor}
\def\BibTeX{{\rm B\kern-.05em{\sc i\kern-.025em b}\kern-.08em
    T\kern-.1667em\lower.7ex\hbox{E}\kern-.125emX}}
\begin{document}


\title{On Decentralized Sum-Rate Maximization with Successive Interference Cancellation
\thanks{This work was supported by the Comunidad Aut\'onoma de Madrid Atracci\'on de Talento grant 2020-T1/TIC-20698 and the Generaci\'on de Conocimiento grant PID2022-142506NA-I00.}
}

\author{\IEEEauthorblockN{David Garrido}
\IEEEauthorblockA{\textit{Signal Theory and Communications} \\
\textit{Universidad Carlos III de Madrid}\\
Leganes, Madrid \\
davgarri@ing.uc3m.es}
\and
\IEEEauthorblockN{Marcos M. Vasconcelos}
\IEEEauthorblockA{\textit{Electrical and Computer Engineering} \\
\textit{Florida State University}\\
Tallahassee, Florida \\
marcos@eng.famu.fsu.edu}
\and
\IEEEauthorblockN{Borja Peleato}
\IEEEauthorblockA{\textit{Signal Theory and Communications} \\
\textit{Universidad Carlos III de Madrid}\\
Leganes, Spain \\
bpeleato@ing.uc3m.es}
}

\maketitle

\begin{abstract}
Successive Interference Cancellation (SIC) is a powerful technique for managing interference in wireless networks, yet its optimal deployment in decentralized environments remains a challenge.
This study investigates joint power and rate allocation in a two-user Gaussian interference channel incorporating SIC at the receivers.
We characterize the global optimal solutions of the problem, and recognizing the limitations of centralized coordination, we introduce a novel decentralized algorithm for a symmetric channel configuration.
Numerical results demonstrate that even without global Channel State Information, our proposed algorithm significantly outperforms traditional benchmarks, such as Orthogonal Access which suffers from temporal underutilization or greedy strategies that fail to exploit SIC gains.
\end{abstract}


\section{Introduction}

Interference is a fundamental limitation in multi-user communication systems, particularly in wireless networks where time and spectrum are inherently scarce resources. Among the many approaches proposed to improve spectral efficiency, Non-Orthogonal Multiple Access (NOMA) has attracted significant attention by allowing multiple users to share the same time–frequency resources~\cite{saito2013non}. In its most common form, NOMA relies on Superposition Coding (SC) at the transmitter and Successive Interference Cancellation (SIC) at the receiver. Nevertheless, alternative NOMA schemes have been proposed that do not rely on SC \cite{jafarkhani2024modulation} or SIC \cite{chung2021correlated}. 

A central challenge in NOMA systems lies in selecting appropriate powers and rate allocations for each signal. Although
numerous algorithms have been proposed for resource allocation in interference-limited networks (e.g., \cite{naderializadeh2022state,tabrizi2015spatial,garrido2023resource,10888922}), most assume independent transmitter–receiver pairs and simply treat interference as noise. As a result, they do not explicitly account for the possibility of interference cancellation at the receivers~\cite{chandra2024analyzing} -- an approach that could substantially improve the performance of the system. 

From an information-theoretic perspective, the capacity region of the general interference channel remains unknown, even for the two-user Gaussian case. This motivates the study of practical transmission strategies that exploit additional receiver-side capabilities~\cite{zeng2017sum,ding2015application,sun2015ergodic}. In this work, we investigate joint power and rate allocation in the two-user Gaussian interference channel while explicitly incorporating SIC at the receivers. Unlike most prior approaches, our framework allows receivers to opportunistically decode and cancel interference when beneficial, enabling operating points beyond the conventional ``treat-interference-as-noise'' regime~\cite{li2018maximum,trankatwar2024power}.

We focus on a simplified yet insightful setting: the two-by-two Gaussian interference channel in which the cross-channel gains are weaker than the direct gains. For this model, we derive structural properties of the optimal (centralized) power and rate allocation when SIC is available. Leveraging these properties, we propose a low-complexity heuristic algorithm for adaptive power and rate control that can be implemented in a decentralized manner. We evaluate its performance in terms of sum-rate maximization and compare it against both the centralized optimal solution and several baseline schemes. Numerical results show that the proposed algorithm significantly outperforms selfish (greedy) rate maximization and (traditional) orthogonal access strategies based on time or bandwidth splitting.

\section{System model}

\begin{figure}
    \centering
    \includegraphics[width=0.4\linewidth]{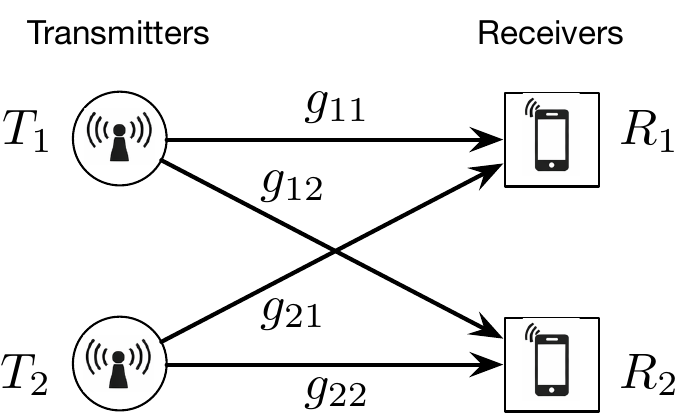}
    \caption{Interference channel model considered in this paper.}
    \label{fig:two_by_two_system}
\end{figure}

We consider the system shown in Fig.~\ref{fig:two_by_two_system}, with two transmitters $T_1, T_2$ and two receivers $R_1, R_2$, where $T_i$ transmits to $R_i$, $i=1,2$.
The signal received by $R_i$ at time $t$ is given by
\begin{equation}
    y_i(t) = \sqrt{g_{1i}} x_1(t) + \sqrt{g_{2i}} x_2(t) + z_i(t),
\end{equation}
where $x_j(t)$ is the signal transmitted by $T_j$ at time $t$ (if any), $g_{ji}$ the (power) path loss from $T_j$ to $R_i$, and $z_i(t)$ the additive Gaussian noise at $R_i$. 
For simplicity, we assume that the noise is independent and identically distributed with zero mean and variance $\sigma^2$ for all receiver nodes, i.e., $z_i(t) \sim \mathcal{N}(0,\sigma)$, $i\in\{1,2\}$, for all $t$. We neglect propagation delays. 

The capacity in bits per second per Hz [bps/Hz] of a band-limited wireless link treating interference as noise (without SIC) is given by~\cite{cover1999elements}: 
\begin{equation} \label{eq:cap}
    c_{ij} = \log_2 \left( 1 + \frac{g_{ij} \gamma_i}{g_{kj}\gamma_k + 1} \right),
\end{equation}
$i,j,k \in \{1, 2\}$ and $k \neq i$. 
$c_{ij}$ denotes the capacity at receiver $R_j$ for the signal transmitted by $T_i$ and $\gamma_i$ represents the Signal-to-Noise Ratio (SNR) for transmitter $i$. For a transmit power of $P_i$, we have $\gamma_i = P_i / \sigma^2$. Conversely, a receiver successfully cancelling the interference can achieve a rate of
\begin{equation} \label{eq:cap_sic}
    c_{ii} = \log_2 \left( 1 + g_{ii} \gamma_i \right).
\end{equation}

A receiver can decode and cancel an interfering signal if the rate of that signal is below the capacity in Eq.~(\ref{eq:cap}). When that happens, the desired transmitter can use a higher rate, ideally reaching its interference-free capacity in Eq.~(\ref{eq:cap_sic}). Each transmitter is subject to a peak power constraint $P_i \leq P_i^{max}$, which induces a maximum SNR $\gamma_i^{max} = P_i^{max} / \sigma^2$.

Let $r_i$ denote the data rate of transmitter $T_i$. This rate is fundamentally limited by the link capacity, which depends on whether SIC is employed. The total system throughput, or sum-rate, is defined as $r_{sum} = r_1 + r_2$.

Throughout this paper, we employ the superscript $(\cdot)^\star$ to denote the optimal value of a given parameter. Additionally, $\mathbb{E} [\cdot]$ represents the expectation operator, and we define the auxiliary function $\phi(x)=\log_2(1+x)$ to simplify the expressions.

\section{Problem formulation and Sum-Rate Optimal Rate-Power Allocation Strategies} \label{sec:prob_form}

In this paper, the main focus is on the optimal decoding strategies and power allocation for the interference channel of Fig.~\ref{fig:two_by_two_system}.
Let $\gamma_i \in [0,\gamma_i^{\max}]$ denote the transmit power of user
$i\in\{1,2\}$, and let $r_i$ denote its achievable rate.
The sum-rate maximization problem is
\begin{equation}\label{prob:sum_rate_global}
\begin{aligned}
\underset{\mathcal{S},\, r_1,r_2,\, \gamma_1,\gamma_2}{\text{maximize}}\quad
& r_1 + r_2 \\
\text{subject to}\quad
& 0 \le \gamma_1 \le \gamma_1^{\max}, \\
& 0 \le \gamma_2 \le \gamma_2^{\max}, \\
& (r_1,r_2) \in \mathcal{R}\!\left(\mathcal{S}; \gamma_1,\gamma_2\right), \\
& \mathcal{S} \in \{\text{No-SIC},\, \text{Partial-SIC},\, \text{Full-SIC}\},
\end{aligned}
\end{equation}
where $\mathcal{R}(\mathcal{S};\gamma_1,\gamma_2)$ denotes the achievable
rate region induced by the decoding architecture $\mathcal{S}$.

To solve the problem in Eq.~\eqref{prob:sum_rate_global}, we decompose the problem of maximizing the total rate $r_{sum}$ into three sub-problems and select the best solution among them:

\begin{enumerate}
\item \textbf{No-SIC:}
Without SIC, the problem is a standard sum-rate maximization treating interference as noise:
\begin{equation} \label{prob:noSIC}
\begin{array}{ll}
\underset{\gamma_1,\gamma_2}{\text{maximize}}
 & \phi \left(\frac{g_{11}\gamma_1}{g_{21}\gamma_2 + 1}\right) + \phi \left(\frac{g_{22}\gamma_2}{g_{12}\gamma_1 + 1}\right) \\
\text{subject to}
&  0 \leq \gamma_i \leq \gamma_i^{max}, \quad i=1,2
\end{array}
\end{equation}
A solution can be found with 
\begin{equation}
\gamma_1^\star\in\{0,\gamma_1^{\max}\},\qquad \gamma_2^\star\in\{0,\gamma_2^{\max}\}.
\end{equation}

\begin{IEEEproof}
First observe that, since $\phi(x)=\log_2(1+x)$, maximizing the objective is equivalent to maximizing
\begin{equation}
    \log_2 (g_{21} \gamma_2 + 1 + g_{11} \gamma_1 ) - \log_2 (g_{21} \gamma_2 + 1 ) + \log_2 (g_{12} \gamma_1 + 1 + g_{22} \gamma_2) - \log_2 (g_{12} \gamma_1 + 1 )\label{eq:objective_general}
\end{equation}
The partial derivates are 
\begin{align}
    \partial_{\gamma_1} &= \frac{g_{11}}{g_{21} \gamma_2 + 1 + g_{11} \gamma_1}
- \frac{g_{12} g_{22} \gamma_2}{( g_{12} \gamma_1 + 1 ) ( g_{12} \gamma_1 + 1 + g_{22} \gamma_2 )},\label{eq:partial_gamma1_general}\\
   \partial_{\gamma_2} &= \frac{g_{22}}{g_{12} \gamma_1 + 1 + g_{22} \gamma_2}
- \frac{g_{21} g_{11} \gamma_1}{( g_{21} \gamma_2 + 1 ) ( g_{21} \gamma_2 + 1 + g_{11} \gamma_1 )}.\label{eq:partial_gamma2_general}
\end{align}

Since the $\log_2$ is a differentiable function, its maximum (if one exists) must fulfil $\partial_{\gamma_1}=\partial_{\gamma_2}=0$. Expanding the above expressions, equating them to zero, and cancelling the denominators yields:
\begin{align}
    g_{11} ( g_{12} \gamma_1 + 1) (g_{12} \gamma_1 + 1 + g_{22} \gamma_2) = g_{12} g_{22} \gamma_2 (g_{21} \gamma_2 + 1 + g_{11} \gamma_1 ),\\
    g_{22} ( g_{21} \gamma_2 + 1) (g_{21} \gamma_2 + 1 + g_{11} \gamma_1) = g_{21} g_{11} \gamma_1 (g_{12} \gamma_1 + 1 + g_{22} \gamma_2 ).
\end{align}
After some calculations, the above equations are equivalent to
\begin{align}
    g_{11} ( g_{12} \gamma_1 + 1 ) ^2 &= g_{12} g_{22} \gamma_2 (g_{12} \gamma_2 + 1 + g_{11} \gamma_1 ) - g_{11} g_{22} \gamma_2 (g_{12} \gamma_1 + 1)  \\
    &= g_{12} g_{22} \gamma_2 (g_{21} \gamma_2 + 1) - g_{11} g_{22} \gamma_2 \nonumber ,\\
    g_{22} ( g_{21} \gamma_2 + 1 ) ^2 &= g_{21} g_{11} \gamma_1 (g_{12} \gamma_1 + 1) - g_{22} g_{11} \gamma_1 .
\end{align}

Let
\begin{align}
A & = g_{12} \gamma_1 + 1\\ 
B & = g_{21} \gamma_2 + 1,
\end{align}
then the above equations are equivalent to
\begin{align}
    g_{11} g_{21} A ^2 = (B - 1)( g_{12} g_{22} B - g_{11} g_{22} ),\label{eq:intermediate1_general}\\
    g_{22} g_{12} B ^2 = (A - 1)( g_{21} g_{11} A - g_{22} g_{11} ).\label{eq:intermediate2_general}
\end{align}
Eqs.~(\ref{eq:intermediate1_general})~and~(\ref{eq:intermediate2_general}) form a system of equations. Adding them to solve it yields
\begin{equation}
    ( g_{21} + g_{22} ) g_{11} A + (g_{12} + g_{11}) g_{22} B - 2 g_{11} g_{22} = 0
\end{equation}
or equivalent 
\begin{equation}
    A = \frac{2 g_{11} g_{22} - (g_{12} + g_{11}) g_{22} B}{( g_{21} + g_{22} ) g_{11} }.
\end{equation}
Since $ g_{12} > 0$ and $g_{21} >0$, it implies $B>1$ and $A>1$. If $B>1$, then
$A < \frac{2 g_{11} g_{22} - (g_{12} + g_{11}) g_{22} }{( g_{21} + g_{22} ) g_{11} }$, rearranging terms
\begin{equation}
    A < \frac{(g_{11} - g_{12})}{g_{11}} \frac{g_{22}}{(g_{22} + g_{21})}.
\end{equation}
Both fractions are smaller than one, then $A < 1$.
Therefore, there exists no solution for Eqs.~(\ref{eq:intermediate1_general})~and~(\ref{eq:intermediate2_general}) with $A>1$ and $B>1$. Equivalently, there exist no $\gamma_1>0$ and $\gamma_2>0$ for which the partial derivatives in Eqs.~(\ref{eq:partial_gamma1_general})~and~(\ref{eq:partial_gamma2_general}) are zero. As a consequence, there exists no local maxima for Eq.~(\ref{eq:objective_general}) with $\gamma_1>0$ and $\gamma_2>0$ and the solution for problem~(\ref{prob:noSIC}) must be at the boundary.

\end{IEEEproof}

\item \textbf{Partial-SIC:}
Only one receiver performs interference cancellation. For this to be feasible, the other signal must be decodable by both receivers. Assuming, without loss of generality (wlog), that $R_2$ is the one performing SIC, the problem becomes:
\begin{equation} \label{eq:gen_max_prob_partial}
\begin{array}{ll}
\underset{r_1,\gamma_1,\gamma_2}{\text{maximize}}
& r_1 + \phi \left( g_{22}\gamma_2 \right) \\
\text{subject to}
& r_1 \leq \min \left\{ \phi \left( \frac{g_{12}\gamma_1}{g_{22}\gamma_2 + 1}\right) , \phi \left( \frac{g_{11}\gamma_1}{g_{21}\gamma_2 + 1}\right) \right\} \\
& 0 \leq \gamma_i \leq \gamma_i^{max}, \quad i=1,2.
\end{array}
\end{equation}
The first constraint ensures that $r_1$ is decodable at $R_2$ (for cancellation) and at $R_1$ (for information recovery). A solution can be found with 
\begin{equation}
\gamma_1^\star = \gamma_1^{\max},\qquad \gamma_2^\star\in\{0,\gamma_2^{\max}\}.
\end{equation}
\begin{IEEEproof}
    The $\min \{ \cdot \}$ expression yields two possible cases based on the values of
\begin{equation} \label{eq:gamma_thres}
    k_t = \frac{g_{11}-g_{12}}{g_{12}g_{21}-g_{11}g_{22}}.
\end{equation}

In the first case, $\gamma_2 > k_t$, then $\phi \left( \frac{g_{12}\gamma_1}{g_{22}\gamma_2 + 1}\right) < \phi \left( \frac{g_{11}\gamma_1}{g_{21}\gamma_2 + 1}\right)$ and the total rate is 
\begin{equation} \label{eq:partial_1}
    r_{sum}
    = \phi \left( \frac{g_{12}\gamma_1}{g_{22}\gamma_2 + 1}\right)  + \phi \left(1 + g_{22}\gamma_2 \right)
\end{equation}
Since the second term in Eq.~(\ref{eq:partial_1}) is independent of $\gamma_1$, the first term is maximized at $\gamma_1 = \gamma_1^{max}$.
With $\gamma_1$ fixed, Eq.~(\ref{eq:partial_1}) becomes a function of a single variable. Its derivative with respect to $\gamma_2$ is

\begin{equation}
    \frac{dr_{sum}}{d\gamma_2} = \frac{g_{22}}{g_{22} \gamma_2+1}-\frac{g_{12} g_{22} \gamma_1}{(g_{22} \gamma_2+1)^2 + (g_{22} \gamma_2+1)g_{12} \gamma_1 }.
\end{equation}
This derivative never equals zero for any $\gamma_2 >0$, therefore Eq.~(\ref{eq:partial_1}) is monotonically increasing, and the maximum is attained at $\gamma_2 = \gamma_2^{max}$.

In the second case, where $\gamma_2 < k_t$, then $\phi \left(\frac{g_{12}\gamma_1}{g_{22}\gamma_2 + 1}\right) > \phi \left(1 + \frac{g_{11}\gamma_1}{g_{21}\gamma_2 + 1}\right)$ and the total rate is given by
\begin{equation} \label{eq:par_sic_2}
    r_{sum} 
    = \phi \left( \frac{g_{11}\gamma_1}{g_{21}\gamma_2 + 1}\right)  + \phi \left( g_{22}\gamma_2 \right).
\end{equation}
As in the previous case, $\gamma_1 = \gamma_1^{max}$ is optimal and, differentiating with respect to $\gamma_2$ to find the critical points, we obtain:
\begin{equation}
    \frac{dr_{sum}}{d\gamma_2} = \frac{g_{22}}{g_{22} \gamma_2+1} - \frac{g_{21}g_{11}\gamma_1}{(g_{21}\gamma_2+1)^2 + ( g_{21}\gamma_2+1)g_{11}\gamma_1}.
\end{equation}
The derivative equals zero at $\gamma_2 = \frac{1}{g_{21}} \left[- 1 \pm \sqrt{\frac{\gamma_1g_{11}}{g_{22}} (g_{21}-g_{22})} \right]$. The negative solution lies outside the constrain of the problem. Consequently, there is only one valid critical point at $\gamma_2 = \frac{1}{g_{21}} \left[ \sqrt{\frac{\gamma_1g_{11}}{g_{22}} (g_{21}-g_{22})} - 1 \right]$ which is a minimum.

To prove that this critical point is a minimum, note that $g_{21}>g_{22}$ is required for the critical point to be real-valued. Under this condition, the sum-rate function $r_{sum}(\gamma_2)$ defined in Eq.~(\ref{eq:par_sic_2}) has a domain $\gamma_2\in(- \frac{1}{g_{21}},\infty)$. 

Given $\lim_{\gamma_2 \rightarrow -\frac{1}{g_{21}}} r_t = \infty$ and $\lim_{\gamma_2 \rightarrow \infty} r_t = \infty$, 
since there is only on critical point, that point must be a global minimum.
Therefore, the maximum of Eq.~(\ref{eq:par_sic_2}) must occur at the boundaries $\gamma_2 = 0$ or $\gamma_2 = \gamma_2^{max}$.
Combining both cases, the optimal power allocation for the Partial-SIC scenario, is $\gamma_1 = \gamma_1^{max}$ and $\gamma_2 = \{ 0,\gamma_2^{max} \}$.

\end{IEEEproof}

\item \textbf{Full-SIC:}
When both users perform SIC, the sum-rate maximization problem is
\begin{equation} \label{prob:fullsic1}
\begin{array}{ll}
\underset{r_1,r_2,\gamma_1,\gamma_2}{\text{maximize}}
& r_1 + r_2 \\
\text{subject to}
& r_1 \leq \min \left\{  \phi \left( \frac{g_{12}\gamma_1}{g_{22}\gamma_2 + 1}\right) , \phi \left(g_{11} \gamma_1 \right) \right\}, \\
& r_2 \leq \min \left\{  \phi \left( \frac{g_{21}\gamma_2}{g_{11}\gamma_1 + 1}\right) , \phi \left( g_{22} \gamma_2 \right) \right\}, \\
& 0 \leq \gamma_i \leq \gamma_i^{max}, i=1,2.
\end{array}
\end{equation}
A solution can be found with
\begin{equation} \label{eq:global_sol}
\begin{array}{ll}
     &  \gamma_1^\star = \left\{ 0 ,  \frac{g_{21}-g_{22}}{g_{11}g_{22}} , \gamma_1^{max} \right\}, \\
     & \gamma_2^\star = \left\{ 0 , \frac{g_{12}-g_{11}}{g_{11}g_{22}} , \gamma_2^{max} \right\}.
\end{array}
\end{equation}

\end{enumerate}

\begin{IEEEproof}
The arguments of the $\min \{ \cdot \}$ expressions attain equality at $\gamma_1 = k_1$ and $\gamma_2=k_2$, where $k_1 = \frac{g_{21}-g_{22}}{g_{11}g_{22}}$ and $k_2 = \frac{g_{12}-g_{11}}{g_{11}g_{22}}$.
These thresholds define four regions, each leading to a different rate configuration and a corresponding maximization sub-problem. The global optimum for the Full-SIC case is the maximum among the solutions to these four sub-problems.
The four regions are determined by combinations of the following expressions:
\begin{subequations}
    \begin{equation} \label{eq:full_1}
    \text{If } \gamma_1 > k_1, \quad \min \left\{  \phi \left( \frac{g_{21}\gamma_2}{g_{11}\gamma_1 + 1}\right) , \phi \left( g_{22} \gamma_2 \right) \right\} = \phi \left( \frac{g_{21}\gamma_2}{g_{11}\gamma_1 + 1}\right) 
    \end{equation}
    \begin{equation} \label{eq:full_2}
        \text{If } \gamma_1 < k_1 , \quad \min \left\{  \phi \left( \frac{g_{21}\gamma_2}{g_{11}\gamma_1 + 1}\right) , \phi \left(  g_{22} \gamma_2 \right) \right\} = \phi \left( g_{22} \gamma_2 \right)
    \end{equation}
    \begin{equation} \label{eq:full_3}
        \text{If } \gamma_2 > k_2 , \quad \min \left\{  \phi \left( \frac{g_{12}\gamma_1}{g_{22}\gamma_2 + 1}\right) , \phi \left( g_{11} \gamma_1 \right) \right\} = \phi \left( \frac{g_{12}\gamma_1}{g_{22}\gamma_2 + 1}\right)
    \end{equation}
    \begin{equation} \label{eq:full_4}
        \text{If } \gamma_2 < k_2 , \quad \min \left\{  \phi \left( \frac{g_{12}\gamma_1}{g_{22}\gamma_2 + 1}\right) , \phi \left( g_{11} \gamma_1 \right) \right\} = \phi \left( g_{11} \gamma_1 \right)
    \end{equation}    
\end{subequations}

\textbf{Case 1:} Combining Eq.~(\ref{eq:full_2}) and Eq.~(\ref{eq:full_3})

\begin{equation} \label{eq:sub_full_1}
\begin{array}{ll}
\underset{\gamma_1,\gamma_2}{\text{maximize}}
& \phi \left( \frac{g_{12}\gamma_1}{g_{22}\gamma_2 + 1}\right) + \phi \left( g_{22} \gamma_2 \right) \\
\text{subject to}
& 0 \leq \gamma_1 \leq \frac{g_{21}-g_{22}}{g_{11}g_{22}}, \\
& \max \left\{ 0 , \frac{g_{12}-g_{11}}{g_{11}g_{22}} \right\} \leq \gamma_2 \leq \gamma_2^{max}.
\end{array}
\end{equation}
The objective function is identical to that of Eq.~(\ref{eq:partial_1}) in the Partial-SIC scenario; therefore, the solution follows the same logic:
$\gamma_1^* = \frac{g_{21}-g_{22}}{g_{11}g_{22}}$ and $\gamma_2^\star = \gamma_2^{max} $. 

\textbf{Case 2:} Combining Eq.~(\ref{eq:full_1}) and Eq.~(\ref{eq:full_4})

By symmetry, this combination yields the same structure as (\ref{eq:sub_full_1}) with indices exchanged. The resulting solution is $\gamma_1^* = \gamma_1^{max}$ and $\gamma_2^\star = \frac{g_{12}-g_{11}}{g_{11}g_{22}} $.

\textbf{Case 3:} Combining Eq.~(\ref{eq:full_2}) and Eq.~(\ref{eq:full_4})

\begin{equation}
\begin{array}{ll}
\underset{\gamma_1,\gamma_2}{\text{maximize}}
& \phi \left( g_{11} \gamma_1 \right) + \phi \left( g_{22} \gamma_2 \right) \\
\text{subject to}
& 0 \leq \gamma_1 \leq \frac{g_{21}-g_{22}}{g_{11}g_{22}}, \\
& 0 \leq \gamma_2 \leq \frac{g_{12}-g_{11}}{g_{11}g_{22}}.
\end{array}
\end{equation}
The objective function is monotonically increasing in both variables, making the solution trivial: $\gamma_1^\star = \frac{g_{21}-g_{22}}{g_{11}g_{22}}$ and $\gamma_2^\star = \frac{g_{12}-g_{11}}{g_{11}g_{22}}$.

\textbf{Case 4:} Combining Eq.~(\ref{eq:full_1}) and Eq.~(\ref{eq:full_3})

\begin{equation} \label{prob:fullsic4}
\begin{array}{ll}
\underset{\gamma_1,\gamma_2}{\text{maximize}}
& \phi \left( \frac{g_{12}\gamma_1}{g_{22}\gamma_2 + 1}\right) + \phi \left( \frac{g_{21}\gamma_2}{g_{11}\gamma_1 + 1}\right) \\
\text{subject to}
& \max \left\{0, \frac{g_{21}-g_{22}}{g_{11}g_{22}} \right\} \leq \gamma_1 \leq \gamma_1^{max} , \\
& \max \left\{0, \frac{g_{12}-g_{11}}{g_{11}g_{22}} \right\} \leq \gamma_2 \leq \gamma_2^{max}.
\end{array}
\end{equation}

The structure of (\ref{prob:fullsic4}) is identical to that of problem (\ref{prob:noSIC}). Consequently, the same methodology proves that the optimum lies on the boundary:
$\gamma_1^\star \in \{\max \left\{0, \frac{g_{21}-g_{22}}{g_{11}g_{22}} \right\},\gamma_1^{\max}\}$, $ \gamma_2^\star\in \{\max \left\{0, \frac{g_{12}-g_{11}}{g_{11}g_{22}} \right\},\gamma_2^{\max}\}$.
Combining the candidate solutions from these four cases yields the global optimal solution in Eq.~(\ref{eq:global_sol}).

\end{IEEEproof}

Ultimately, the global optimal solution is the maximum among the solutions to these three cases.
Note that Eq.~(\ref{eq:global_sol}) contains the solutions for the other two cases, so it includes the global solution.

Henceforth, we assume that each transmitter is situated closer to its intended receiver than to any other receiver in the network. This is a common assumption for cellular networks, since mobile nodes often connect to their nearest base station, making intended link gains naturally dominant over cross-link interference. Mathematically, this corresponds to the condition:
\begin{equation} \label{eq:condition}
    g_{ii} > g_{ij},
\end{equation}
with $i,j \in \{1, 2\}$ and $i \neq j$.

\begin{proposition}
If Eq.~(\ref{eq:condition}) holds, the maximum sum-rate without interference cancellation $r^\star_{NS}$, is always greater than or equal to that with full cancellation $r^\star_{FS}$:
\begin{equation}
    r_{NS}^\star \geq r_{FS}^\star.
\end{equation}
Furthermore, the optimum occurs when both transmitters operate at their maximum available power, i.e., 
\begin{equation}
        \gamma_i^\star = \gamma_i^{max},\qquad i=1,2.
\end{equation}
\end{proposition}
\begin{IEEEproof}
To prove the first statement, note that under condition (\ref{eq:condition}), the $\min\{\cdot\}$ expressions in (\ref{prob:fullsic1}) always yield the first term, therefore, problem (\ref{prob:fullsic1}) reduces to problem (\ref{prob:fullsic4}).
Comparing the first term of the objective functions of problems (\ref{prob:noSIC}) and (\ref{prob:fullsic4}) reveals that the former is always greater than or equal to the latter, thus ensuring $r_{NS}^\star \geq r_{FS}^\star$.

For the second statement, given the previous result, we only need to compare the No-SIC and Partial-SIC scenarios.
Under condition (\ref{eq:condition}), Partial-SIC implies $k_t < 0$. Consequently, only Eq.~(\ref{eq:partial_1}) needs to be maximized, and its solution is $\gamma^\star_i = \gamma_i^{max}$ for $i=1,2$.
Evaluating the objective function of problem (\ref{prob:noSIC}) at $\gamma_1 = 0$ and $\gamma_2 = \gamma_2^{max}$ and comparing it with Eq.~(\ref{eq:partial_1}) at $\gamma_i = \gamma_i^{max}$, it is clear that the latter is greater. Since boundary cases where $\gamma_1 = 0$ or $\gamma_2 = 0$ are not optimal for the global problem, the optimum must occur at $\gamma^\star_i = \gamma_i^{max}$ for $i=1,2$.

\end{IEEEproof}

\begin{proposition} \label{pro:noisy}
    If Eq.~(\ref{eq:condition}) holds and the noise power is sufficiently large ($\gamma_i \ll 1$), the No-SIC strategy strictly dominates the Partial-SIC strategies:
    \begin{equation}
        r^\star_{NS} > \max \{ r^\star_{PI}, r^\star_{PII} \},
    \end{equation}
    where $r^\star_{PI}$ and $r^\star_{PII}$ are the sum-rates following a Partial-SIC strategy when $R_2$ or $R_1$ remove the interference respectively.
\end{proposition}
\begin{IEEEproof}
Equating the objective function of problem (\ref{prob:noSIC}) to Eq.~(\ref{eq:partial_1}) and solving for $\gamma_2$ yields $\gamma_2 = \frac{g_{11} - g_{12}}{g_{21}g_{12}}$. This result defines the configurations where No-SIC outperforms Partial-SIC:
\begin{equation} \label{condition:1}
    \text{If} \quad \gamma_2 < \frac{g_{11} - g_{12}}{g_{21}g_{12}}, \quad r_{NS} > r_{PI}
\end{equation}
Similarly, when $R_1$ performs SIC, the condition becomes
\begin{equation} \label{condition:2}
    \text{if} \quad \gamma_1 < \frac{g_{22} - g_{21}}{g_{12}g_{21}}, \quad r_{NS} > r_{PII}.
\end{equation}
Consequently, in the limit as $\gamma_i \rightarrow 0$ for $i=1,2$, inequalities (\ref{condition:1}) and  (\ref{condition:2}) are always satisfied. Therefore, $r_{NS} > \max \{r_{PI} ,r_{PII} \}$ holds in the high-noise regime.

\end{IEEEproof}

\begin{proposition} \label{pro:0noise}
    If the noise power is sufficiently small ($\gamma_i \rightarrow \infty$), the optimal sum-rate $r_{sum}^\star$ approaches the capacity of a single transmitter
    \begin{equation}
        r_{sum}^\star \simeq \max \{\phi( g_{11} \gamma_1^{max}) , \phi( g_{22} \gamma_2^{max})\}.
    \end{equation}
\end{proposition}

\begin{IEEEproof}
Let $\gamma_1, \gamma_2 \rightarrow \infty$, and and for simplicity, represent both by the same variable $\gamma$. As $\gamma \rightarrow \infty$, Full-SIC reduces to problem (\ref{prob:fullsic4}) because all other sub-problems become unfeasible. Furthermore, the rate for Full-SIC remains bounded: 
\begin{equation}
    \lim_{\gamma \rightarrow \infty} \left[  \phi \left( \frac{g_{12}\gamma}{g_{22}\gamma + 1}\right) + \phi \left( \frac{g_{21}\gamma}{g_{11}\gamma + 1}\right) \right] < \infty,
\end{equation}
Similarly, for No-SIC:
\begin{equation}
     \lim_{\gamma \rightarrow \infty} \left[ \phi \left( \frac{g_{11}\gamma}{g_{21}\gamma + 1}\right) + \phi \left( \frac{g_{22}\gamma}{g_{12}\gamma + 1}\right) \right] < \infty.
\end{equation}
In contrast, Partial-SIC reduces to Eq.(\ref{eq:partial_1}), which provides the optimal solution because it is unbounded in the limit:
\begin{equation}
    \lim_{\gamma \rightarrow \infty} \left[ \phi \left( \frac{g_{12}\gamma_1}{g_{22}\gamma_2 + 1}\right)  + \phi \left( g_{22}\gamma_2 \right) \right] = \infty.
\end{equation}
A similar behaviour is observed for a single transmitter, which has a rate $r_{single} = \phi \left( g_{22}\gamma_2 \right)$, where
\begin{equation}
    \lim_{\gamma \rightarrow \infty}  r_{single} = \infty.
\end{equation}
To compare these expressions, we define $\Omega$ as $\Omega = \frac{  \phi \left( \frac{g_{12}}{g_{22}}\right)}{\phi \left( \frac{g_{12}\gamma}{g_{22}\gamma + 1}\right)  + \phi \left( g_{22}\gamma \right)}$. It follows that
\begin{equation}
    r_{single} \geq r^\star_{PI} (1 - \Omega).
\end{equation}
In the limit as$\gamma \rightarrow \infty$, $\Omega \rightarrow 0$, demonstrating that a single-transmitter strategy is quasi-optimal, i.e., $r_{single} \simeq r^\star_{PI}$.

\end{IEEEproof}

\begin{remark}[High-SNR and Low-SNR regimes] Propositions \ref{pro:noisy} and \ref{pro:0noise} characterize the dominant transmission strategies at extreme SNR scales. In the low-SNR regime ($\gamma_i \ll 1$), where noise is the primary limiting factor, treating interference as noise (No-SIC) is the optimal policy as the power of the interfering signal is negligible compared to the noise floor. Conversely, in the high-SNR regime ($\gamma_i \gg 1$), the system becomes interference-limited; in this state, silencing one user makes a quasi-optimal approach.
\end{remark}

The utility of Successive Interference Cancellation becomes clear in the intermediate SNR regime. In this transition zone, neither noise nor interference is sufficiently dominant to justify the simpler strategies of greedy or orthogonal transmissions. It is within these intermediate settings that SIC serves as a critical bridge, allowing the system to decode and subtract interfering signals to recover significant throughput that would otherwise be lost to interference or underutilized time-slots.

From a centralized perspective, a network controller with global Channel State Information (CSI) can determine the optimal sum-rate by evaluating the three sub-problems and coordinating the transmitters accordingly. However, in practical deployments, a centralized controller is often unavailable, or the signalling overhead required for coordination may be prohibitive. This motivates a decentralized algorithmic approach to jointly allocating rate and power for the interference channel.
Decentralizing the rate selection for arbitrary channel gains is a challenging problem due to the large number of possible configurations. 

To develop fundamental insights that could be useful to more general interference networks, we restrict our analysis to a symmetric normalized setting as follows. 
We assume that transmitters are equidistant from their intended receivers, such that $g_{11}=g_{22}=1$. To model the dominance of the intended link over interference, the cross-link gains are defined as $g_{21} = 1 - \epsilon$ and $g_{12} = 1 - \mu$, with $\epsilon,\mu \in (0,1)$ representing the margin by which the intended link outperforms the interfering ones.

Furthermore, we assume identical peak SNR for both transmitters. To improve readability, the corresponding sub- and super-indices will be omitted, denoting $\gamma = \gamma_1^{max} = \gamma_2^{max}$.

Under these assumptions, the global optimization problem reduces to selecting the maximum among three candidate sum-rates:
\begin{equation} \label{eq:max_reduced}
    \max \{ r^\star_{NS}, r^\star_{PI}, r^\star_{PII} \}
\end{equation}
with 
\begin{align}
    r^\star_{NS} &= \phi \left( \frac{\gamma}{(1-\epsilon) \gamma + 1}\right) + \phi \left( \frac{\gamma}{(1 - \mu) \gamma + 1}\right), \label{eq:NoSIC}\\
    r^\star_{PI} &= \phi \left( \frac{(1 - \mu) \gamma}{\gamma + 1}\right)  + \phi \left( \gamma \right),\label{eq:PartialI} \\
    r^\star_{PII} &= \phi \left( \gamma \right) + \phi \left( \frac{(1-\epsilon) \gamma}{\gamma + 1}\right). \label{eq:PartialII}
\end{align}

The choice between the two Partial-SIC configurations depends on the relative values of $\mu$ and $\epsilon$.
If $\mu > \epsilon$, it follows that Eq. (\ref{eq:PartialI}) $<$ Eq. (\ref{eq:PartialII}).
Wlog, we assume $\mu \geq \epsilon$; the case where $\mu<\epsilon$ is perfectly symmetric by exchanging the user indices.
The primary comparison, therefore, lies between $r^\star_{NS}$ and $r^\star_{PII}$. These two expressions are equal along a critical boundary defined by the following switching curve:
\begin{equation} \label{eq:curve}
    \mu = \frac{ \gamma (\epsilon-1) }{(\epsilon-1) \gamma - 1}.
\end{equation}
This curve, plotted in Fig.~\ref{fig:h_shape}, delineates the regions in the ($\epsilon,\mu$) plane where Partial-SIC outperforms the No-SIC strategy.

Figure \ref{fig:h_shape} illustrates the optimal strategy for a representative SNR of $\gamma = 4$. 
The diagonal $\mu = \epsilon$ serves as the boundary between the two Partial-SIC configurations. For each maximum SNR $\gamma$, this line intersects the switching curve in Eq.~\eqref{eq:curve} at $\epsilon=\mu=q(\gamma)$ with
\begin{equation} \label{eq:borders_meet}
    q (\gamma) = \frac{1 + 2 \gamma - \sqrt{1 + 4 \gamma}}{2 \gamma}.
\end{equation}
As $\gamma$ increases, this intersection point shifts, expanding the region where Partial-SIC is the optimal strategy.

\begin{figure}
\centering
\includegraphics[width=.45\textwidth]{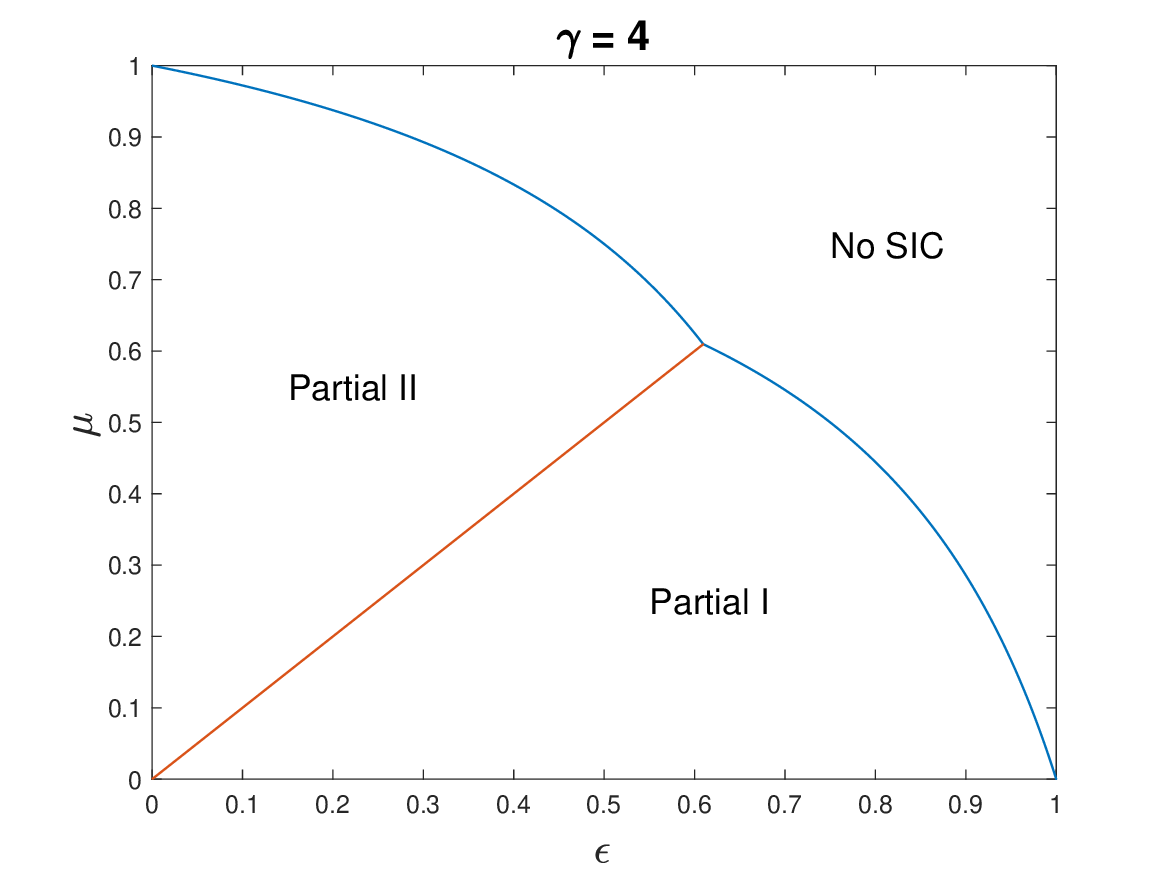}
\caption{Optimal solution for equidistant intended links.}
\label{fig:h_shape}
\end{figure}

\section{Decentralized algorithm}


The proposed algorithm employs a rate-oscillation mechanism for one transmitter, while the other greedily transmits at the highest feasible rate and performs SIC whenever possible. Since the transmitters do not communicate, this oscillatory behavior enables one transmitter to create favorable conditions for the other to perform SIC, thereby achieving a higher sum rate than strategies that completely avoid interference cancellation.
The algorithm operates in two distinct phases: an initialization phase where both transmitters determine their respective roles and a steady-state phase during which one transmitter alternates between two rate values, while the other varies its rate between zero and the capacity.

\begin{figure}
\centering
\includegraphics[width=.45\textwidth]{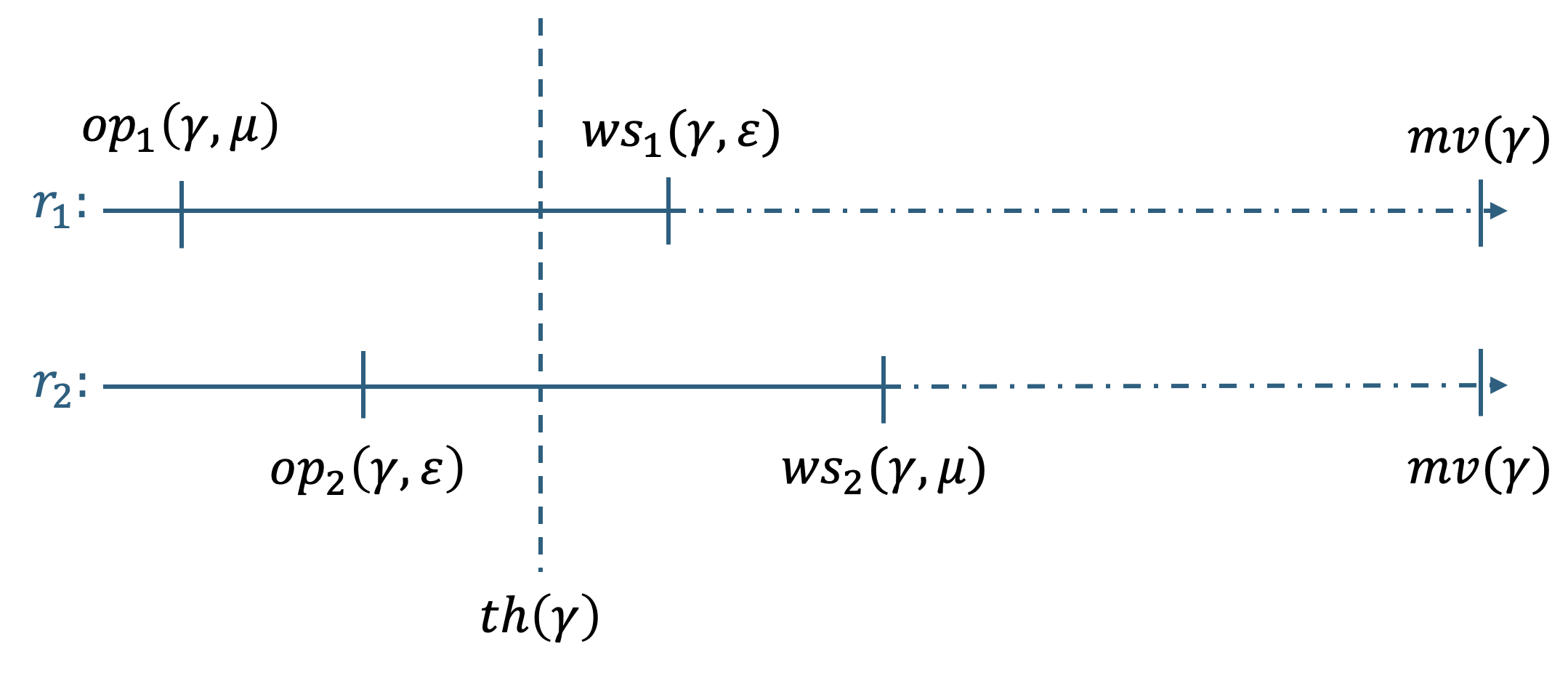}
\caption{Illustration of the key rate values for the case $\mu > \epsilon$.}
\label{fig:algorithm}
\end{figure}

Figure~\ref{fig:algorithm} illustrates the key operating points of the algorithm, assuming that $\mu \geq \epsilon$, wlog. 
Transmission rates increase from left to right, ranging from 0 to the maximum value $mv(\gamma) = \phi(\gamma)$.
The rates at which signals can be decoded by their intended receiver without SIC are denoted by $ws_1(\gamma,\epsilon) = \phi \left( \frac{\gamma}{(1-\epsilon)\gamma + 1}\right)$ and $ws_2(\gamma,\mu) = \phi \left( \frac{\gamma}{(1-\mu)\gamma + 1}\right)$. Signals with rates below those values are always feasible, while those with rates on the horizontal dashed line are only feasible when SIC is successfully performed.
The rates at which interferences can be decoded by the other receiver are denoted 
$op_1(\gamma,\mu) = \phi \left( \frac{(1-\mu)\gamma}{\gamma + 1}\right)$ and $op_2(\gamma,\epsilon) = \phi \left( \frac{(1-\epsilon)\gamma}{\gamma + 1}\right)$. Interferences with rates below those values can always be cancelled, while those with rates beyond the vertical dashed line $th(\gamma) = \phi \left(  \frac{\gamma}{\gamma + 1} \right) $ can never be cancelled, regardless of $\epsilon$ and $\mu$. 

The algorithm proceeds as follows. Both transmitters simultaneously start transmitting at the maximum rate $r_i = mv(\gamma)$ and then gradually decrease their rates at the same speed until SIC becomes possible or the rate reaches 0. Initially, neither receiver can decode either signal but $T_2$ eventually reaches the rate $r_2 = ws_2(\gamma,\mu)$ and $R_2$ becomes able to decode its signal. Shortly thereafter, $R_1$ also becomes capable of decoding its intended signal. When $r_2$ reaches $op_2(\gamma,\epsilon)$, $R_1$ can perform SIC and $T_1$ switches to the maximum rate $r_1 = mv(\gamma)$. Eventually, $T_2$ reaches $r_2 = 0$ and, since $R_2$ never managed to cancel interference, $T_2$ infers that $T_1$ is transmitting at the highest possible rate. This marks the end of phase one, roles are assigned. 

During the steady-state phase, the centralized solution for $T_2$ (whose receiver never achieved SIC) is $r_2=op_2(\gamma,\epsilon)$. However, $T_2$ does not know $\epsilon$, so it varies $r_2$ periodically within the interval $r_2 \in \left[ 0 , ws_2(\gamma,\mu) \right]$, with period $T$. Specifically, at time $t$ we have
\begin{equation}
r_2 = \left\{\begin{array}{lr}
v\cdot (t\ \mathrm{mod}\ T) & v\cdot (t\ \mathrm{mod}\ T)\leq th(\gamma)\\
ws_2(\gamma) & v\cdot (t\ \mathrm{mod}\ T)>th(\gamma)
\end{array}\right.
\end{equation}
with
\begin{equation} \label{eq:speed}
    v = \frac{1}{T} \phi \left( \frac{\gamma}{(1 - \mu)\gamma + 1}\right).
\end{equation}

In contrast, $T_1$ spends the steady-state phase greedily transmitting at the maximum rate tolerated by $R_1$, using SIC when possible. When $r_2\geq op_2(\gamma,\epsilon)$, $T_1$ reduces its rate to $r_1 = ws_1(\gamma,\epsilon)$ to ensure decodability and once $r_2$ falls below this threshold, $T_1$ resumes transmission at the maximum rate $r_1=mv(\gamma)$.

\begin{figure}
    \centering
    \includegraphics[width=.5\linewidth]{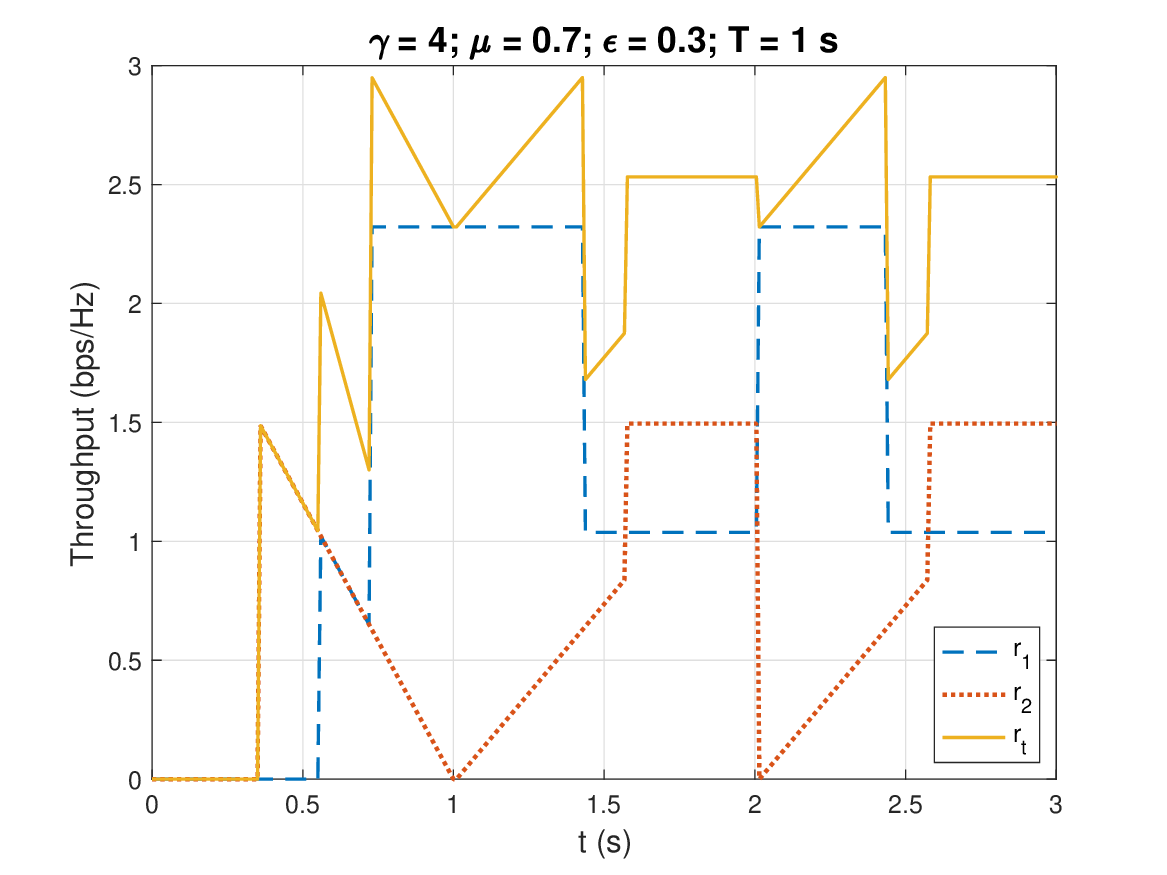} 
    \caption{Algorithm evolution for $\gamma = 4$, period $T=1$, $\mu = 0.7$ and $\epsilon = 0.3$.}
    \label{fig:algoritm_evol}
\end{figure}

\begin{example} \label{ex:algorithm}
To further illustrate the algorithm’s behaviour, Fig.~\ref{fig:algoritm_evol} depicts the temporal evolution of the transmission rates for a particular example with $\epsilon=0.3$, $\mu=0.7$, $\gamma=4$ and period $T=1$~s.
During the initialization phase (from $t=0$ to $t=1~s$), both transmitters initially operate at rates above capacity, preventing successful decoding. As time progresses, their rates decrease steadily. At $t=0.36$~s, $T_2$ reaches the rate $r_2 = ws_2(4,0.7) = 1.49$~bps/Hz, enabling $R_2$ to decode its signal. Receiver $R_1$ achieves successful decoding at $t=0.55$~s. Shortly afterward, when $r_2$ reaches $op_2(4,0.3) = 0.64$~bps/Hz, $R_1$ can cancel interference and $T_1$ switches to the maximum rate.

The steady-state phase begins at $t=1$~s. $R_1$ removes the interference until $t=1.44$~s when the $r_2 > 0.64$~bps/Hz and SIC cannot be performed. $T_2$ continues increasing the rate without knowing the changes of $r_1$ until it reaches $r_2 = th(4) = 0.85$~bps/Hz. 
The rate that allows $R_1$ to remove the interference is below $0.85$, transmitting with a rate different from the capacity for the time it is above it is useless, therefore, $T_2$ changes the rate to the greatest possible $r_2 = 1.49$~bps/Hz.
\end{example}

The proposed algorithm assumes that transmitters possess knowledge of both intended channels ($g_{ii} = 1$) and, despite its decentralized nature, it requires a certain degree of coordination during the initial phase. Specifically, both transmitters must decrease their rates simultaneously and at the same pace. While the rate of decrease (i.e., the speed) can be agreed upon in advance, some form of signaling between users is still necessary to synchronize their timing.

While the algorithm is sensitive to the specific values of $\epsilon$ and $\mu$, it is designed to ensure robust performance on average. In the absence of global Channel State Information, selecting a fixed transmission strategy is likely to be suboptimal. By hybridizing the No-SIC and Partial-SIC schemes, the algorithm mitigates the inherent limitations of each individual approach, leveraging their respective advantages to maintain a high average sum-rate across diverse channel configurations.

To compare the performance of the algorithm proposed with other procedures we need to compute the expected sum rate $\mathbb{E} [r_{os} (t) ]$,
where $r_{os}(t)$ is the total rate $r_{sum}$ at a particular time $t$ for transmitters following the proposed algorithm.
We ignore the initial phase of the algorithm for simplicity since, for sufficiently long time periods, it can be neglected.


\begin{theorem}
The expected total rate of the algorithm is given by
\begin{equation} \label{eq:exp_tot_rate}
    \mathbb{E} [r_{os} (t) ] = \mathbb{E} [ r_1 (t) ] + \mathbb{E} [ r_2 (t) ] 
\end{equation}
where
\begin{align} 
    \mathbb{E} [ r_1 (t) ] & = \frac{\phi \left( \frac{(1-\epsilon)\gamma}{\gamma + 1}\right) }{\phi \left( \frac{\gamma}{(1 - \mu)\gamma + 1}\right)}   \left[ \phi( \gamma) -  \phi \left( \frac{\gamma}{(1 - \epsilon)\gamma + 1}\right) \right] + \phi \left( \frac{\gamma}{(1 - \epsilon)\gamma + 1}\right)\label{eq:expected_1},\\
    \mathbb{E} [r_2(t)] & = \frac{1}{2} \frac{\phi^2 \left( \frac{\gamma}{\gamma + 1}\right)}{\phi \left( \frac{\gamma}{(1 - \mu)\gamma + 1}\right)} + \phi \left( \frac{\gamma}{(1 - \mu)\gamma + 1}\right) - \phi \left( \frac{\gamma}{\gamma + 1}\right).\label{eq:expected_2}
\end{align}
\end{theorem}
\begin{IEEEproof}
The expected value of the algorithm is
\begin{equation} \label{eq:r_os_proof}
    \mathbb{E} [r_{os} (t) ] = \mathbb{E} [ r_1 (t) ] + \mathbb{E} [ r_2 (t) ] = \frac{1}{T} \int_0^T r_1(t) dt + \frac{1}{T} \int_0^T r_2 (t) dt
\end{equation}
where $r_1(t)$ and $r_2(t)$ are the are the instantaneous rates of each transmitter:
\begin{equation} \label{eq:r_1(t)}
r_1 (t)= \left\{\begin{array}{lr}
\phi ( \gamma) & t \leq t'\\
\phi \left( \frac{\gamma}{(1 - \epsilon)\gamma + 1}\right) & t > t'
\end{array}\right.
\end{equation}

\begin{equation} \label{eq:r_2(t)}
r_2(t) = \left\{\begin{array}{lr}
v \cdot t & t \leq t''\\
 \phi \left( \frac{\gamma}{(1 - \mu)\gamma + 1}\right) & t > t''
\end{array}\right.
\end{equation}
In these expressions, $t'$ denotes the time when $r_2(t') = \phi \left( \frac{(1-\epsilon)\gamma}{\gamma + 1}\right)$, marking the point beyond which interference cancellation can no longer be performed at $R_1$. Similarly, $t''$ is the time when $r_2 (t'')= \phi \left(  \frac{\gamma}{\gamma + 1} \right) $. 
At $t > t''$ the rate required for $R_1$ to remove interference lies below the current transmission rate. Since transmitting at a rate below capacity is suboptimal once this threshold is exceeded, transmitter $T_2$ switches its rate to the maximum feasible value.
Using the speed defined in Eq.~(\ref{eq:speed}) and setting $r_2(t') = v \cdot t'$, the time $t'$ is claculated as 
\begin{equation} \label{eq:tprime}
    t' = \frac{r_2(t')}{v} = \frac{T \phi \left( \frac{(1-\epsilon)\gamma}{\gamma + 1}\right)}{\phi \left( \frac{\gamma}{(1 - \mu)\gamma + 1}\right)}.
\end{equation}
Similarly, $t''$ is given by,
\begin{equation} \label{eq:t2prime}
    t'' = \frac{r_2(t'')}{v} = \frac{T \phi \left(  \frac{\gamma}{\gamma + 1} \right)}{\phi \left( \frac{\gamma}{(1 - \mu)\gamma + 1}\right)}.
\end{equation}

The expected values of the individual rates in Eq.~(\ref{eq:r_1(t)}) and Eq.~(\ref{eq:r_2(t)}) are
\begin{equation} \label{eq:gen_expected_1}
\begin{split}
    \mathbb{E} [ r_1 (t) ] & = \frac{1}{T}  \left[  \int_0^{t'} \phi ( \gamma) dt +   \int_{t'}^T \phi \left( \frac{\gamma}{(1 - \epsilon)\gamma + 1}\right) dt  \right]\\
    & = \frac{t'}{T} \phi ( \gamma) + \frac{(T-t')}{T} \phi \left( \frac{\gamma}{(1 - \epsilon)\gamma + 1}\right),
\end{split}
\end{equation}

\begin{equation} \label{eq:gen_expected_2}
\begin{split}
    \mathbb{E} [ r_2 (t) ] & = \frac{1}{T}  \left[  \int_0^{t''} \frac{1}{T} \phi \left( \frac{\gamma}{(1 - \mu)\gamma + 1}\right) dt +   \int_{t''}^T \phi \left( \frac{\gamma}{(1 - \mu)\gamma + 1}\right) dt  \right]\\
    & = \frac{t''^2}{2T^2} \phi \left( \frac{\gamma}{(1 - \mu)\gamma + 1}\right) + \frac{(T-t'')}{T} \phi \left( \frac{\gamma}{(1 - \mu)\gamma + 1}\right).
\end{split}
\end{equation}
Substituting Eq.~(\ref{eq:tprime}) and Eq.~(\ref{eq:t2prime}) in Eq.~(\ref{eq:gen_expected_1}) and Eq.~(\ref{eq:gen_expected_2}) respectively, yields the expected values in Eq.~(\ref{eq:expected_1}) and Eq.~(\ref{eq:expected_2}).

\end{IEEEproof}



\section{Empirical Results}

We define the efficiency $\rho$ of an algorithm as the ratio of its expected total rate and the maximum theoretical rate from Eq.~(\ref{eq:max_reduced}). Specifically, for our algorithm
\begin{equation}
    \rho_{os} = \frac{\mathbb{E} [r_{os} (t) ] }{\max \{ r^\star_{NS}, r^\star_{PI}, r^\star_{PII} \}},
\end{equation}
where $\mathbb{E} [r_{os} (t) ]$ is given in Eq.~(\ref{eq:exp_tot_rate}). Figure \ref{fig:effi_4} illustrates this efficiency as a function of $\epsilon$ and $\mu$ when the maximum SNR is $\gamma=4$. A prominent ridge of peak efficiency is visible, which marks the boundary between the No-SIC and the Partial-SIC regions. Since both strategies offer similar performance in this area, the sub-optimal cyclic oscillations of our algorithm do not have a significant impact on the total rate. 
Conversely, as the channel parameters move away from this boundary, the optimality gap of our algorithm widens. 

\begin{figure}[h]
    \centering
    \includegraphics[width=.5\linewidth]{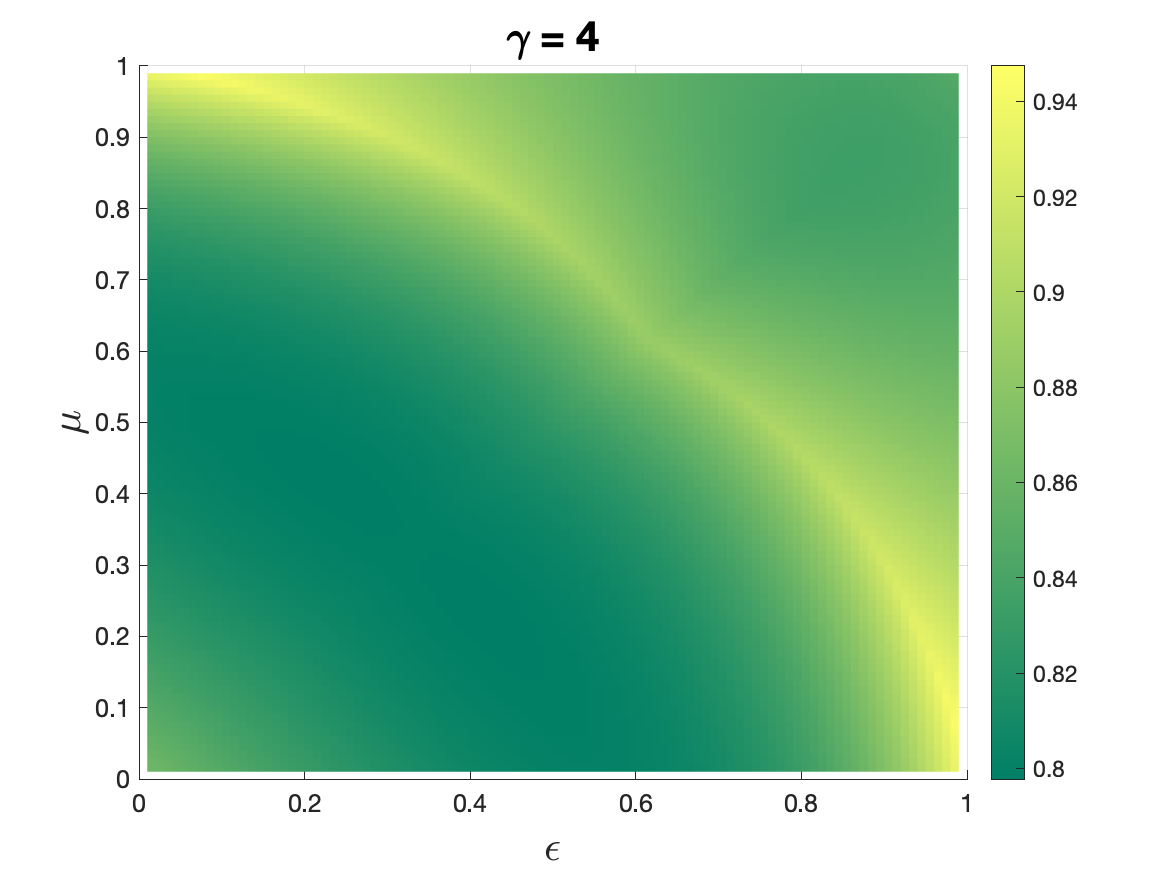} 
    \caption{Efficiency of the proposed algorithm for $\gamma = 4$.}
    \label{fig:effi_4}
\end{figure}

To further validate the proposed algorithm, we compare its performance against two benchmarks:
\begin{itemize}
\item Greedy Strategy: Both transmitters operate at their peak power and the maximum rate supported by their respective receivers, disregarding the interference caused to the other link. This is equivalent to the No-SIC scenario, hence the expected total rate is given by
\begin{equation}
    \mathbb{E} [r_{gr}(t)] = r^\star_{NS}.
\end{equation}

\item Orthogonal Access: 
Each transmitter is active during half of the total time interval, utilizing maximum power and rate without mutual interference.

Since both transmitters share the same peak SNR $\gamma$, the expected total rate is 
\begin{equation}
    \mathbb{E} [r_{or} (t) ] =  \phi ( \gamma).
\end{equation}

\end{itemize}
Figure~\ref{fig:compare} illustrates the efficiency of the three approaches for $\mu=0.2$ and varying $\epsilon$. The dotted vertical line denotes the transition between the Partial-SIC and No-SIC optimal regions.
As expected, the Greedy Strategy is optimal within the No-SIC region. However, its efficiency degrades significantly in Partial-SIC regions, where interference cancellation is critical.

On the other hand, Orthogonal Access is equivalent to having only one active transmitter. This approach is fundamentally suboptimal compared to SIC-enabled strategies. 
Performing SIC allows one user to transmit at the maximum rate while the other maintains a strictly positive (albeit reduced) rate, ensuring that the sum-rate always exceeds the orthogonal limit.
Nevertheless, orthogonal access 
is not universally inferior to the proposed algorithm. For instance, Prop.~\ref{pro:noisy} shows that when $\gamma \gg 1$, orthogonal access is nearly optimal. In such scenarios, it may surpass the performance of the proposed algorithm; the hybrid nature of our approach forces the system to spend time in a No-SIC state, which can be less efficient than pure time-sharing.

The proposed algorithm provides structural robustness by functioning as a dynamic combination of No-SIC and Partial-SIC strategies, ensuring consistent performance across diverse channel configurations.
Unlike the Greedy approach, which fails to exploit the throughput gains of interference cancellation, or Orthogonal Access, which is limited by time-slot partitioning, our decentralized oscillating mechanism dynamically adapts to the channel conditions.
In a decentralized environment where transmitters lack global Channel State Information, a static strategy risks severe efficiency loss.

\begin{figure}[h]
    \centering
    \includegraphics[width=.5\linewidth]{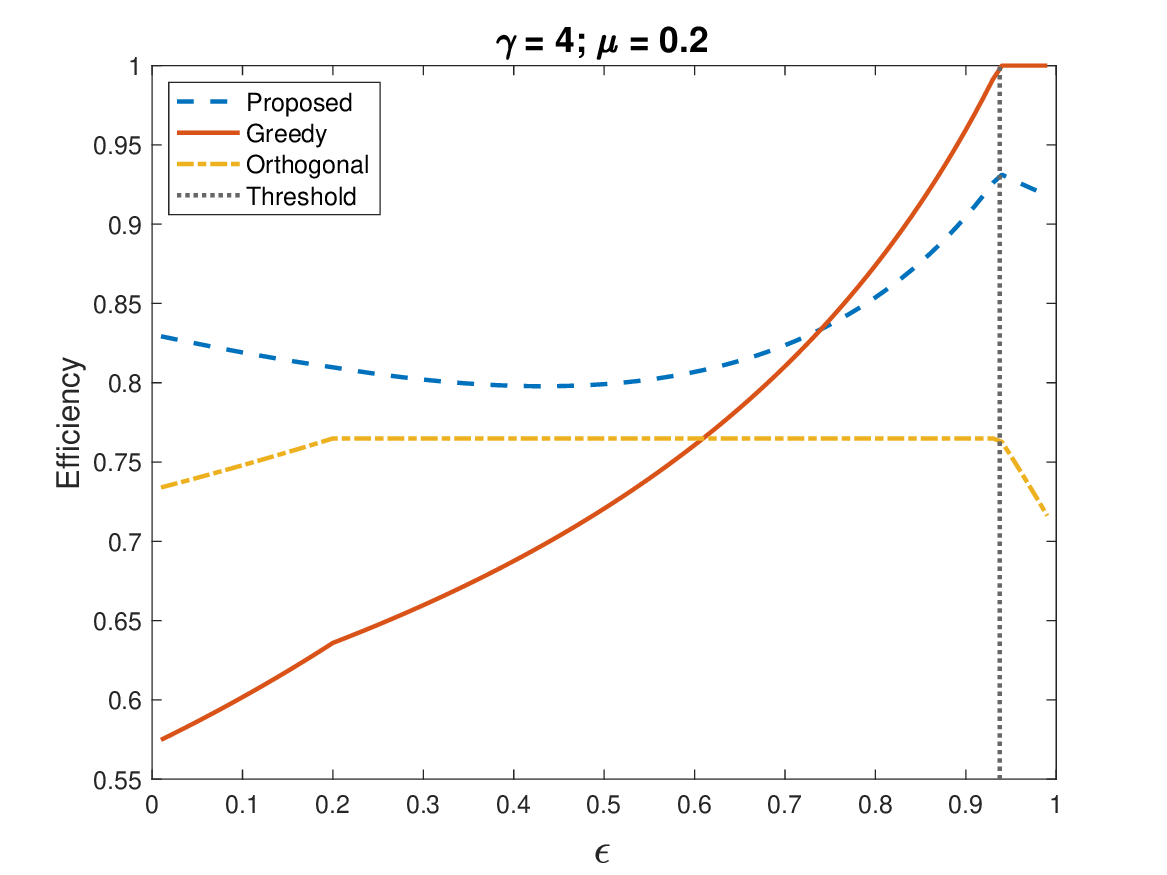} 
    \caption{Efficiency of the 3 algorithms with $\gamma = 4$.}
    \label{fig:compare}
\end{figure}

\section{Conclusions and future work}

In this work, we investigated the joint power and rate allocation problem for a two-user Gaussian interference channel, explicitly incorporating the capabilities of Successive Interference Cancellation at the receivers. We characterized the global optimal solution for arbitrary channel configurations and, acknowledging the practical constraints of decentralized networks, proposed a decentralized algorithm for the symmetric case where transmitters are equidistant from their respective receivers.

Our performance analysis demonstrates that the proposed procedure outperforms two conventional benchmarks: the Greedy Strategy, which suffers from its inability to exploit SIC gains, and Orthogonal Access, which is limited by time-slot partitioning. By dynamically adapting to the interference environment, our approach achieves high efficiency without requiring global Channel State Information.

Future research will focus on two primary objectives. First, we aim to generalize the decentralized framework to n-user networks, addressing the increased complexity of SIC ordering. Second, we intend to optimize the internal parameters of the oscillating mechanism, specifically investigating non-linear transition curves to further enhance system throughput.

\bibliographystyle{IEEEtran}
\bibliography{bibliography}

@article{naderializadeh2022state,
  title={State-augmented learnable algorithms for resource management in wireless networks},
  author={NaderiAlizadeh, Navid and Eisen, Mark and Ribeiro, Alejandro},
  journal={IEEE Transactions on Signal Processing},
  volume={70},
  pages={5898--5912},
  year={2022},
  publisher={IEEE}
}

@article{garrido2023resource,
  title={A Resource Allocation Algorithm for Collaborative Networks Using Inferred Information},
  author={Garrido, David and Zhang, Mai and Peleato, Borja},
  journal={IEEE Access},
  volume={11},
  pages={34685--34697},
  year={2023},
  publisher={IEEE}
}

@article{li2018maximum,
  title={Maximum sum rate of slotted aloha with successive interference cancellation},
  author={Li, Yitong and Dai, Lin},
  journal={IEEE Transactions on Communications},
  volume={66},
  number={11},
  pages={5385--5400},
  year={2018},
  publisher={IEEE}
}

@inproceedings{chandra2024analyzing,
  title={Analyzing the Performance of Various NOMA Systems Based on User Count and Minimum Rate Demand},
  author={Chandra, K Ramesh and Basamsetti, Anusha and Swathi, Kankipati and Kumar, NV Phani Sai and Durga, Vendra Bhavani and Raju, B Elisha},
  booktitle={2023 4th International Conference on Intelligent Technologies (CONIT)},
  pages={1--5},
  year={2024},
  organization={IEEE}
}

@article{chung2021correlated,
  title={Correlated superposition coding: Lossless two-user NOMA implementation without SIC under user-fairness},
  author={Chung, Kyuhyuk},
  journal={IEEE Wireless Communications Letters},
  volume={10},
  number={9},
  pages={1999--2003},
  year={2021},
  publisher={IEEE}
}

@inproceedings{trankatwar2024power,
  title={Power Allocation for Sum Rate Maximization Under SIC Constraint in NOMA Networks},
  author={Trankatwar, Sachin and Wali, Prashant K},
  booktitle={2024 16th International Conference on COMmunication Systems \& NETworkS (COMSNETS)},
  pages={646--650},
  year={2024},
  organization={IEEE}
}

@article{tabrizi2015spatial,
  title={Spatial reuse in dense wireless areas: A cross-layer optimization approach via ADMM},
  author={Tabrizi, Haleh and Peleato, Borja and Farhadi, Golnaz and Cioffi, John M and Aldabbagh, Ghadah},
  journal={IEEE Transactions on Wireless Communications},
  volume={14},
  number={12},
  pages={7083--7095},
  year={2015},
  publisher={IEEE}
}

@article{jafarkhani2024modulation,
  title={Modulation and Coding for NOMA and RSMA},
  author={Jafarkhani, Hamid and Maleki, Hossein and Vaezi, Mojtaba},
  journal={Proceedings of the IEEE},
  year={2024},
  publisher={IEEE}
}

@inproceedings{saito2013non,
  title={Non-orthogonal multiple access (NOMA) for cellular future radio access},
  author={Saito, Yuya and Kishiyama, Yoshihisa and Benjebbour, Anass and Nakamura, Takehiro and Li, Anxin and Higuchi, Kenichi},
  booktitle={2013 IEEE 77th vehicular technology conference (VTC Spring)},
  pages={1--5},
  year={2013},
  organization={IEEE}
}

@article{zeng2017sum,
  title={On the sum rate of MIMO-NOMA and MIMO-OMA systems},
  author={Zeng, Ming and Yadav, Animesh and Dobre, Octavia A and Tsiropoulos, Georgios I and Poor, H Vincent},
  journal={IEEE Wireless communications letters},
  volume={6},
  number={4},
  pages={534--537},
  year={2017},
  publisher={IEEE}
}

@article{ding2015application,
  title={The application of MIMO to non-orthogonal multiple access},
  author={Ding, Zhiguo and Adachi, Fumiyuki and Poor, H Vincent},
  journal={IEEE transactions on wireless communications},
  volume={15},
  number={1},
  pages={537--552},
  year={2015},
  publisher={IEEE}
}

@article{sun2015ergodic,
  title={On the ergodic capacity of MIMO NOMA systems},
  author={Sun, Qi and Han, Shuangfeng and Chin-Lin, I and Pan, Zhengang},
  journal={IEEE Wireless Communications Letters},
  volume={4},
  number={4},
  pages={405--408},
  year={2015},
  publisher={IEEE}
}

@INPROCEEDINGS{10888922,
  author={Bhattacharya, Sagnik and Rajabalifardi, Kamyar and Mohsin, Muhammad Ahmed and Cioffi, John M.},
  booktitle={ICASSP 2025 - 2025 IEEE International Conference on Acoustics, Speech and Signal Processing (ICASSP)}, 
  title={Optimum Power-Subcarrier Allocation and Time-Sharing in Multicarrier NOMA Uplink}, 
  year={2025},
  volume={},
  number={},
  pages={1-5},
  doi={10.1109/ICASSP49660.2025.10888922}}

@book{cover1999elements,
  title={Elements of information theory},
  author={Cover, Thomas M},
  year={1999},
  publisher={John Wiley \& Sons}
}

\end{document}